\documentclass[onecolumn,nonote,noversion]{cdcarticle}

\def\MODE{1}

\def\bugfix{1}

\usepackage{math}
\usepackage[hidelinks]{hyperref}
\usepackage{enumerate}
\usepackage{tikz}
\usepackage{lipsum}

\DeclareMathOperator{\IQC}{IQC}

\begin{document}

\title{Exponential Convergence Bounds using\\ Integral Quadratic Constraints}

\if\MODE1\author{Ross Boczar \and Laurent Lessard \and Benjamin Recht}
\else\author{Ross Boczar\footnotemark[1] \and Laurent Lessard\footnotemark[2] \and Benjamin Recht\footnotemark[1]}\fi

\if\MODE2
\note{}
\else
\note{To Appear, IEEE Conference on Decision and Control, 2015\vspace{-1pt}}
\fi
\maketitle

\if\MODE1\else
\footnotetext[1]{R.~Boczar and B. Recht are with the University of California, Berkeley, CA~94720, USA. \texttt{\{boczar,brecht\}@berkeley.edu}}
\footnotetext[2]{L.~Lessard was with University of California, Berkeley for the duration of this work and is now with the University of Wisconsin--Madison, Madison, WI~53706, USA. \texttt{laurent.lessard@wisc.edu}}
\fi

\begin{abstract}
The theory of integral quadratic constraints (IQCs) allows verification of stability and gain-bound properties of systems containing nonlinear or uncertain elements. Gain bounds often imply exponential stability, but it can be challenging to compute useful numerical bounds on the exponential decay rate.
In this work, we present a modification of the classical IQC results of Megretski and Rantzer~\cite{megretski_system_1997} that leads to a tractable computational procedure for finding exponential rate certificates.
We demonstrate the effectiveness of our method via a numerical example.
\end{abstract}

\section{Introduction}\label{sec:introduction}

In robust control problems, we seek absolute performance guarantees about a system in the presence of bounded uncertainty. Examples of such results include the small gain theorem \& passivity theory~\cite{zames}, dissipativity theory~\cite{willems721}, and integral quadratic constraints (IQCs)~\cite{megretski_system_1997}.

In this paper, we present a modification of IQC theory, the most general of the aforementioned tools, that allows one to certify \emph{exponential stability} rather than just bounded-input bounded-output (BIBO) stability. Moreover, we can compute numerical bounds on the exponential decay rate of the state.

Even when BIBO stable systems are exponentially stable, the estimates of the exponential decay rates provided by standard IQC theory are typically very conservative.
We will show that this conservatism can be greatly reduced if we directly certify exponential stability and use the method presented herein to compute the associated decay rate.

Our modified IQC analysis was successfully applied in~\cite{lessard_analysis_2014} to analyze convergence properties of commonly-used optimization algorithms such as the gradient descent method. These algorithms converge at an exponential rate when applied to strongly convex functions, and the modified IQC analysis automatically produces very tight bounds on the convergence rates.

Another potential application is in time-critical applications such as embedded model predictive control, where it is vital to have robust guarantees that the desired error bounds will be achieved in the allotted time without overflow errors and in spite of fixed-point arithmetic. See \cite{mpc} and references therein.

\paragraph{A special case.}

As previously noted, exponential stability certificates are often conservative when they are derived from $L_2$ gain bounds. However, it is well known that exponential stability can be proven directly in some special cases. To illustrate this fact, consider a discrete linear time-invariant (LTI) plant $G$ with state-space realization $(A,B,C,D)$. Suppose $G$ is connected in feedback with a passive nonlinearity $\Delta$. A sufficient condition for BIBO stability is that there exists a positive definite matrix $P \succ 0$ and a scalar $\lambda \ge 0$ satisfying the linear matrix inequality (LMI)
\begin{equation}\label{eq:passiv}
\addtolength{\arraycolsep}{-0pt}
\bmat{A & B \\ I & 0}^\tp\bmat{P & 0 \\ 0 & -P}\bmat{A & B \\ I & 0}\\
+ \lambda \bmat{0 & C^\tp \\ C & D\!+\!D^\tp} \prec 0
\end{equation}
If we define $V(x) \defeq x^\tp P x$, then~\eqref{eq:passiv} implies that $V$ decreases along trajectories:
$
V(x_{k+1}) \le V(x_k)
$
for all $k$. BIBO stability then follows from positivity and boundedness of $V$.
But observe that when~\eqref{eq:passiv} holds, we may replace the right-hand side by $-\epsilon P$ for some $\epsilon > 0$ sufficiently small. We then conclude that
$
V(x_{k+1}) \le (1-\epsilon)V(x_k)
$
for all $k$ and exponential stability follows. We can then maximize $\epsilon$ subject to feasibility of~\eqref{eq:passiv} to further improve the rate bound.

Unfortunately, the simple trick shown above does not work in the general IQC setting due to the different role played by $P$ in the associated LMI. The LMI used in IQC theory comes from the Kalman-Yakubovich-Popov (KYP) lemma and although it is structurally similar to~\eqref{eq:passiv}, $P$ is not positive definite in general and $V$ may not decrease along trajectories.

Our key insight is that with a suitable modification to both the LMI \emph{and} the IQC definition, we obtain a condition that can certify exponential stability.

The paper is organized as follows. We cover some related work in the remainder of the introduction, we explain our notation and some basic results in Section~\ref{sec:prelim}, we develop and present our main result in Section~\ref{sec:main}, and we discuss computational considerations in~Section~\ref{sec:computation}. Finally, we present an illustrative example demonstrating the usefulness of our result in Section~\ref{sec:examples}, and we make some concluding remarks in Section~\ref{sec:conclusion}.

\paragraph{Related work.}
It is noted in~\cite{megretski_system_1997,rantzer_system_1994} that BIBO stability often implies exponential stability. In particular, we get exponential stability if the nonlinearity satisfies an additional \emph{fading memory} property. The proof of this result is chiefly concerned with showing \emph{existence} of an exponential decay rate. Although the proof  constructs an exponential rate, the construction is based on the assumed $L_2$ gain of the linear map, and thus can be very conservative.

Other proofs of exponential stability have appeared in the literature for specific classes of nonlinearities. Some examples include~\cite{corless_bounded_1993,konishi_robust_1999}, which treat sector-bounded nonlinearities, and \cite{jonsson_nonlinear_1997}, which treats nonlinearities satisfying a Popov IQC. These works exploit LMI modifications akin to the one shown with~\eqref{eq:passiv} earlier in this section.

The sequel is inspired by the recent paper~\cite{lessard_analysis_2014}, which presents an approach for proving the robust exponential stability of optimization algorithms. The approach of~\cite{lessard_analysis_2014} uses a time-domain formulation of IQCs modified to handle exponential stability. In contrast, the present work develops the aforementioned exponential stability modification entirely in the frequency domain and clarifies its connection to the seminal IQC results in~\cite{megretski_system_1997}.  

\section{Notation and preliminaries}\label{sec:prelim}

We adopt a setup analogous to the one used in~\cite{megretski_system_1997}, with the exception that we will work in discrete time rather than continuous time.
%
%
%
The conjugate transpose of a vector $v\in\C^n$ is denoted $v^*$. The unit circle in the complex plane is denoted $\T \defeq \set{z\in\C}{|z|=1}$
The $z$-transform of a time-domain signal $x\defeq (x_0,x_1,\dots)$ is denoted $\hat x (z)$ and defined as
$
\hat x(z) \defeq \sum_{k=0}^\infty x_k z^{-k}
$.

A Hermitian positive definite (semidefinite) matrix $M$ is denoted $M \succ 0$ ($M \succeq 0$).
Function composition is denoted $(g\circ f)(x):=g(f(x))$.
A sequence $u = (u_0,u_1,\dots)$ is said to be in $\ltwo$ if $\sum_{k=0}^\infty|u_k|^2<\infty$. A sequence $u_k$ is said to be in $\ltwo^\rho$ for some $\rho \in (0,1)$ if the sequence $(\rho^{-k}u_k)$ is in $\ltwo$, i.e. $\sum_{k=0}^\infty\rho^{-2k}|u_k|^2<\infty$. Note that $\ltwo^\rho \subset \ltwo$.
%
Let $\RHinf^{m\times n}$ be the set of $m\times n$ matrices whose elements are proper rational functions with real coefficients analytic on the closed unit disk.

Consider the standard setup of Fig.~\ref{fig:lure1}. The block $G$ contains the known LTI part of the system while $\Delta$ contains the part that is uncertain, unknown, nonlinear, or otherwise troublesome.
	\begin{figure}[thpb]
		\centering
		\includegraphics{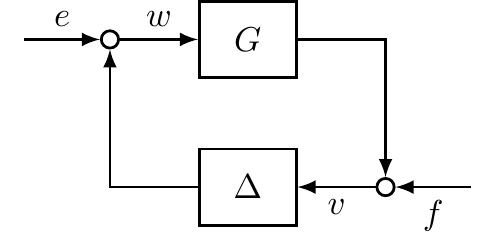}
		\caption{Linear time-invariant system $G$ in feedback with a nonlinearity $\Delta$.
		}\label{fig:lure1}
	\end{figure}
	The interconnection is said to be \emph{well-posed} if the map $(v,w)\mapsto(e,f)$ has a causal inverse. The interconnection is said to be BIBO stable if, in addition, there exists some $\gamma > 0$ such that when $G$ is initialized with zero state,
	\begin{equation*}\label{eq:gain}
	\norm{v}^2 + \norm{w}^2 \le \gamma \bl( \norm{e}^2 + \norm{f}^2 \br)
	\end{equation*}
	for all square-summable inputs $f$ and $e$, and where $\norm{\cdot}$ denotes the $\ltwo$ norm.
	Finally, the interconnection is \emph{exponentially stable} if there exists some $\rho \in (0,1)$ and $c > 0$ such that if $f=0$ and $e=0$, the state $x_k$ of $G$ will decay exponentially with rate $\rho$. That is,
	\[
	\norm{x_k} \le c\, \rho^k\, \norm{x_0}
	\qquad\text{for all }k.
	\]

We now present the classical IQC definition and stability result, which will be modified in the sequel to guarantee exponential convergence. These results are the discrete-time analog of the main IQC results of Megretski and Rantzer~\cite{megretski_system_1997}.

\begin{defn}[IQC]\label{def:iqc}
	Signals $y \in \ltwo$ and $u \in \ltwo$  with associated $z$-transforms  $\hat y(z)$ and $\hat u(z)$  satisfy the \emph{IQC} defined by a Hermitian complex-valued function $\Pi$ if
  \begin{equation}\label{iqceq}
    \int_\T\,
    \bmat{ \hat y(z)\\ \hat u(z) }^* \Pi(z)
    \bmat{ \hat y(z)\\ \hat u(z) } dz \geq 0\:.
  \end{equation}
A bounded operator $\Delta$ satisfies the IQC defined by $\Pi$ if~\eqref{iqceq} holds for all $y\in\ltwo$ with $u = \Delta(y)$. We also define $\IQC(\Pi(z))$ to be the set of all $\Delta$ that satisfy the IQC defined by $\Pi$.
\end{defn}
\begin{thm}[Stability result]
\label{thm:classic}
Let $G(z) \in \RHinf^{m\times n}$ and let $\Delta$ be a bounded causal operator. Suppose that:
  \begin{enumerate}[i)]
    \item for every $\tau \in [0,1]$, the interconnection of $G$ and $\tau \Delta$ is well-posed.
    \item for every $\tau \in [0,1]$, we have $\tau \Delta \in \IQC(\Pi(z))$.
    \item there exists $\epsilon > 0$ such that
      \begin{equation*}
          \begin{bmatrix}
          G(z)\\I
          \end{bmatrix}^*
          \Pi(z)
          \begin{bmatrix}
          G(z)\\I
          \end{bmatrix} \preceq -\epsilon I, \quad \forall z\in \T\:.
      \end{equation*}
  \end{enumerate}
Then, the feedback interconnection of $G$ and $\Delta$ is stable.
\end{thm}

\section{Frequency-domain condition}\label{sec:main}
In this section, we augment Definition~\ref{def:iqc} and the classical result of Theorem~\ref{thm:classic} to derive a frequency-domain condition that certifies exponential stability.

\begin{defn}
The operators $\rho_+,\:\rho_-$ are defined as the time-domain, time-dependent multipliers $\rho^k, \rho^{-k}$, respectively, where $\rho\in(0,1)$ is a defined constant.
\end{defn}
\begin{rem}\label{rem:rho}
The operator $\rho_- \circ (G(z) \circ \rho_+)$ is equivalent to the operator $G(\rho z)$. This follows from the fact that, for any constant $a>0$ and signal $u_k$, the $z$-transform of $a^{-k}u_k$ is given by $\hat u(az)$. See Fig.~\ref{fig:rho_comp} for an illustration.
\begin{figure}[thpb]
  \centering
  \includegraphics{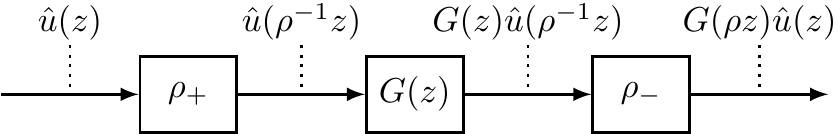}
  \caption{Illustration of Remark \ref{rem:rho}.}
  \label{fig:rho_comp}
\end{figure}
\end{rem}
In order to show exponential stability of the system in Fig.~\ref{fig:lure1}, we will relate it to BIBO stability of the modified system shown in Fig.~\ref{fig:lure2}. This equivalence is closely related to the theory of stability multipliers \cite{safonov_zames-falb_2000}.

\begin{figure}[thpb]
  \centering
  \includegraphics{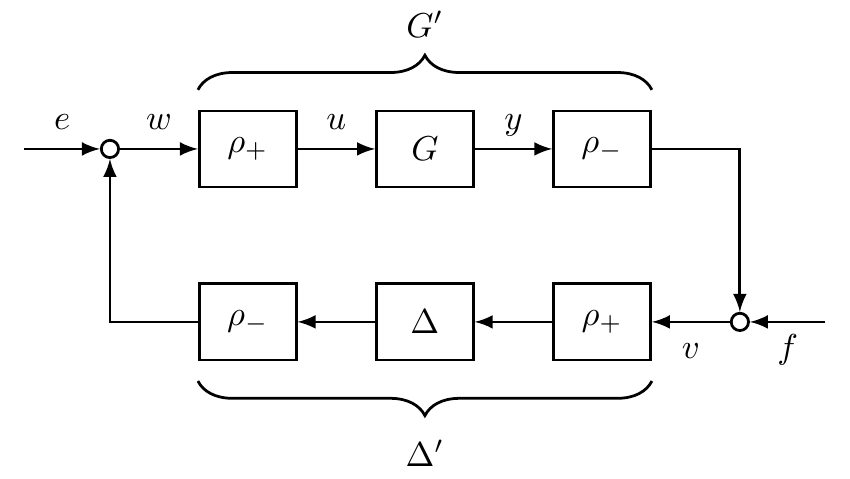}
  \caption{Modified feedback diagram with additional multipliers and inputs. For appropriately chosen $e$ and $f$ and with zero initial condition, we show how this diagram is equivalent to that of Fig.~\ref{fig:lure1}.}
  \label{fig:lure2}
\end{figure}

\begin{prop}\label{prop:exp}
	Suppose $G(z)$ has a minimal realization $(A,B,C,D)$. If the interconnection in Fig.~\ref{fig:lure2} is stable with zero initial condition, then the interconnection in Fig.~\ref{fig:lure1} with initial state~$x_0$ is exponentially stable.
	
\end{prop}
\begin{proof}
Intuitively, if $v$ and $w$ are \emph{small} in the BIBO sense compared to $e$ and $f$, then $y$ must be even smaller. \if\MODE1
A complete proof is included in the appendix.
\else
See~\cite{boczar_arxiv} for a detailed proof.
\fi
	\end{proof}

Ideally, we would like to find a suitable redefinition of the IQC for this transformed system shown in Fig.~\ref{fig:lure2}. To this end, we introduce the concept of the \emph{$\rho$-IQC}.
\begin{defn}[$\rho$-IQC]\label{def:piqc}
	Signals $y \in \ltwo^\rho$ and $u \in \ltwo^\rho$  with associated $z$-transforms  $\hat y(z)$ and $\hat u(z)$  satisfy the \emph{$\rho$-IQC} defined by a Hermitian complex-valued function $\Pi$ if
  \begin{equation}\label{rr}
  \int_\T\,
  \begin{bmatrix} \hat y(\rho z)\\ \hat u(\rho z) \end{bmatrix}^* \Pi(\rho z)
  \begin{bmatrix} \hat y(\rho z)\\ \hat u(\rho z) \end{bmatrix} dz \geq 0\:.
  \end{equation}
	A bounded operator $\Delta$ satisfies the $\rho$-IQC defined by $\Pi$ if~\eqref{rr} holds for all $y\in\ltwo^\rho$ with $u = \Delta(y)$. We also define $\IQC(\Pi(z),\rho)$ to be the set of all $\Delta$ that satisfy the $\rho$-IQC defined by $\Pi$.	
\end{defn}

Note that the concept of a $\rho$-IQC generalizes that of a regular IQC. Indeed, we have $\IQC(\Pi(z),1) = \IQC(\Pi(z))$. The restriction of $u \in \ltwo^\rho$ and $y \in \ltwo^\rho$ corresponds to the restriction of $u \in \ltwo$ and $y \in \ltwo$ in the classical definition of IQC \cite{megretski_system_1997}. 
%
%
Now equipped with $\rho$-IQCs, we can relate $\Delta'$ in Fig.~\ref{fig:lure2} to  $\Delta$ in Fig.~\ref{fig:lure1}.

\begin{prop}\label{prop:iqcequiv}
Let $\Delta$ be a nonlinearity, and let $\Pi$ be a Hermitian complex-valued function. As in Fig.~\ref{fig:lure2}, define $\Delta' \defeq \rho_- \circ (\Delta \circ \rho_+)$. Then the following statements are equivalent.
\begin{enumerate}[(i)]
	\item $\Delta \in \IQC(\Pi(z),\rho)$
	\item $\Delta' \in \IQC(\Pi(\rho z) )$
\end{enumerate}
\end{prop}

\begin{proof}
We define the discrete Fourier transform of the input and output of $\Delta$ as $\hat y(z)$ and $\hat u(z)$, respectively. Then, from the definition of $\rho_+$ and $\rho_-$, we have that $\hat w(z) = \hat u(\rho z)$ and $\hat v(z) = \hat y(\rho z)$. Substituting into the IQC definition~\eqref{iqceq}, we obtain~\eqref{rr} as required.
\end{proof}
Proposition~\ref{prop:iqcequiv} is illustrated in Fig.~\ref{fig:rho_equiv}.
\begin{figure}[thpb]
	\centering
	\includegraphics{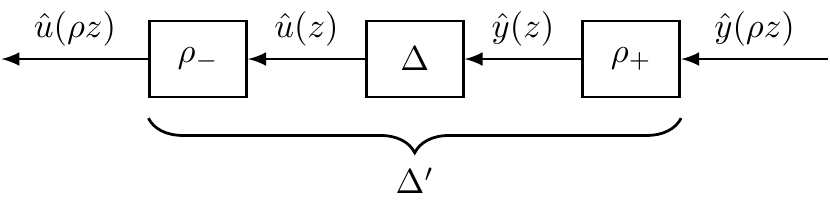}
	\caption{Illustration of Proposition \ref{prop:iqcequiv}.}
	\label{fig:rho_equiv}
\end{figure}

We now state our main result, an exponential stability theorem analogous to the classical result in Theorem~\ref{thm:classic}.

\begin{thm}[Exponential stability]
\label{thm:main}
Fix $\rho \in (0,1)$. Let $G(\rho z) \in \RHinf^{m\times n}$ and let $\Delta$ be a bounded causal operator. Suppose that:
  \begin{enumerate}[i)]
    \item for every $\tau \in [0,1]$, the interconnection of $G$ and $\tau \Delta$ is well-posed.
    \item for every $\tau \in [0,1]$, we have $\tau\Delta \in \IQC(\Pi(z),\rho)$.
    \item there exists $\epsilon > 0$ such that
      \begin{equation}\label{eq:expfdi}
          \begin{bmatrix}
          G(\rho z)\\I
          \end{bmatrix}^*
          \Pi(\rho z)
          \begin{bmatrix}
          G(\rho z)\\I
          \end{bmatrix} \preceq -\epsilon I, \quad \forall z\in \T\:.
      \end{equation}
  \end{enumerate}
Then, the interconnection of $G$ and $\Delta$ shown in Fig.~\ref{fig:lure1} is exponentially stable with rate $\rho$.
\end{thm}

\begin{proof}
Roughly, we apply Theorem \ref{thm:classic} to the interconnection in Fig.~\ref{fig:lure2} with operators $G'$ and $\Delta'$ and IQC $\Pi(\rho z)$.
\if\MODE1
\begin{enumerate} [(a)]
  \item Since Fig.~\ref{fig:lure1} and Fig.~\ref{fig:lure2} have the same interconnection structure, well-posedness is equivalent. 
  \item Due to the equivalence of IQCs in Proposition \ref{prop:iqcequiv},
  \begin{align*}
  &\hspace{-3em}\tau \Delta \in \IQC(\Pi(z),\rho)\\
  &\; \iff \rho_- \circ ((\tau \Delta) \circ \rho_+) \in \IQC(\Pi(\rho z)) \\
  &\; \iff \tau (\rho_- \circ (\Delta \circ \rho_+)) \in \IQC(\Pi(\rho z))  \\
  &\; \iff \tau \Delta' \in \IQC(\Pi(\rho z)) \:.
  \end{align*}
  \item This is condition iii) of Theorem \ref{thm:classic} using $G'$ and $\Delta'$.
\end{enumerate}
Thus, these three conditions ensure BIBO stability of the system in Fig.~\ref{fig:lure2}. We then apply Proposition~\ref{prop:exp} to arrive at exponential stability of Fig.~\ref{fig:lure1}.
\else
See~\cite{boczar_arxiv} for a detailed proof.
\fi
\end{proof}

\section{Computation}\label{sec:computation}

As in the classical IQC setting, to guarantee stability, the frequency-domain inequality (FDI) \eqref{eq:expfdi} must be verified for every $\omega \in [0,2\pi)$. However, if the IQC in question exhibits a \emph{factorizaton}, then the discrete-time KYP Lemma can be applied to convert the infinite-dimensional FDI to a finite-dimensional LMI. We now review these results.

\begin{defn}\label{defn:factor}
We say $\Pi$ has a \emph{factorization} $(\Psi,M)$ if
\begin{equation*}
  \Pi(z) = \Psi(z)^*M \Psi(z)\:,
\end{equation*}
where $\Psi$ is a stable linear time-invariant system, $M$ is a constant Hermitian matrix, and $\Psi(z)^*$ denotes the conjugate transpose of $\Psi(z)$.
\end{defn}

\begin{rem}\label{rem:rhoIQC_time_domain}
If $\Pi(z)$ has a factorization $(\Psi,M)$ and $\Psi(\rho z)$ is stable, then~\eqref{rr} is equivalent to 
  \begin{equation*}
    \sum_{k=0}^\infty \rho^{-2k}z_k^\tp  M z_k \geq 0\:,
\quad\text{where }
    z\defeq \Psi \begin{pmatrix}y\\u\end{pmatrix} \:.
  \end{equation*}
This follows immediately from Parseval's theorem.
\end{rem}

The KYP lemma, stated below, is attributed to Kalman, Yakubovich, and Popov. A simple proof and further references can be found in~\cite{rantzer_KYP}.

\begin{lem}[Discrete-time KYP Lemma]\label{lem:kyp}
	Given matrices $A,B$ and a Hermitian matrix $M$, and assuming $A$ has no eigenvalues on the unit circle, the FDI
	\begin{equation*}
	\begin{bmatrix}
	(zI-A)^{-1}B\\I
	\end{bmatrix}^* M 
	\begin{bmatrix}
	(zI-A)^{-1}B\\I
	\end{bmatrix} \prec 0
	\end{equation*}
	holds for all $z\in\T$ if and only if there exists a solution $P=P^\tp$ and $\lambda \ge 0$ to the LMI
	\begin{equation*}
	\begin{bmatrix}
	A&B\\I&0
	\end{bmatrix}^\tp
	\begin{bmatrix}
	P&0\\0&-P
	\end{bmatrix}
	\begin{bmatrix}
	A&B\\I&0
	\end{bmatrix} + \lambda M \prec 0\:.
	\end{equation*}
\end{lem}

\begin{cor}
Suppose the realization of $G$ is given by $(A,B,C,D)$ and assume $\Pi$ has a factorization $(\Psi,M)$, where the realization of $\Psi$ is given by
\begin{equation*}
\Psi = \left[\begin{array}{c|cc}
   A_\Psi & B_{\Psi_1} & B_{\Psi_2} \\ \hlinet
  C_\Psi & D_{\Psi_1} & D_{\Psi_2}
  \end{array}\right].
\end{equation*}
Then \eqref{eq:expfdi} is equivalent to the existence of $P=P^\tp $ and $\lambda\ge 0$ such that
  \begin{equation}\label{eq:explmi}
      \begin{bmatrix}  
    \hat A^\tp P\hat A-\rho^2 P & \hat A^\tp P\hat B\\ \hat B^\tp P \hat A & \hat B^\tp  P \hat B
    \end{bmatrix}
     + 
    \lambda \begin{bmatrix}
    \hat C^\tp \\ \hat D^\tp 
    \end{bmatrix}
    M
    \begin{bmatrix}
    \hat C & \hat D
    \end{bmatrix}
     \prec 0
  \end{equation}
  where $(\hat A, \hat B, \hat C, \hat D)$ are given by
  \begin{equation*}
  \stsp{\hat A}{\hat B}{\hat C}{\hat D} \defeq
  \left[\begin{array}{cc|c}
    A & 0 & B \\
    B_{\Psi_1} C & A_\Psi & B_{\Psi_2}+B_{\Psi_1}D \\ \hlinet
    D_{\Psi_1}C & C_\Psi & D_{\Psi_2} + D_{\Psi_1}D
  \end{array}\right]
  \end{equation*}
  
\end{cor}
\begin{proof}
\if\MODE1
A similar result is proven in \cite{seiler}, which we repeat here for completeness.
\begin{equation*}
\begin{bmatrix}
      G(z)\\I
      \end{bmatrix}^*
      \Pi(z)
      \begin{bmatrix}
      G(z)\\I
  \end{bmatrix}
=
  \begin{bmatrix}\star \\ \star \end{bmatrix}^*
      \left(
      \begin{bmatrix}
      \hat C^\tp \\ \hat D^\tp 
      \end{bmatrix}
      M
      \begin{bmatrix}
      \hat C & \hat D
      \end{bmatrix}
      \right)
      \begin{bmatrix}
      (zI-\hat A)^{-1}\hat B\\I
      \end{bmatrix}
\end{equation*}
where $\star$ denotes the repeated part of the quadratic form.
Similarly, we have
\begin{align*}
  & \rho^{-2}\begin{bmatrix}
      G(\rho z)\\I
      \end{bmatrix}^*
      \Pi(\rho z)
      \begin{bmatrix}
      G(\rho z)\\I
  \end{bmatrix}\\
  &=
  \begin{bmatrix}\star \\ \star \end{bmatrix}^*
      \left(
      \begin{bmatrix}
      \hat C^\tp \\ \hat D^\tp 
      \end{bmatrix}
      \rho^{-2}M
      \begin{bmatrix}
      \hat C & \hat D
      \end{bmatrix}
      \right)
      \begin{bmatrix}
      (\rho zI-\hat A)^{-1}\hat B\\I
      \end{bmatrix}\\
  &= 
   \begin{bmatrix}\star \\ \star \end{bmatrix}^*
    \left(
    \begin{bmatrix}
    \hat C^\tp \\ \hat D^\tp 
    \end{bmatrix}
    \rho^{-2}M
    \begin{bmatrix}
    \hat C & \hat D
    \end{bmatrix}
    \right)
    \begin{bmatrix}
    (zI-\rho^{-1}\hat A)^{-1}\hat \rho^{-1}B\\I
    \end{bmatrix}
\end{align*}
If $\rho^{-1}\hat A$ has no eigenvalues on the unit circle, we may then invoke Lemma~\ref{lem:kyp} (applied to $\rho^{-1} \hat A$, $\rho^{-1} \hat B$, and the appropriate $M$ term) and multiply through by $\rho^2$ to show that \eqref{eq:expfdi} is equivalent to the existence of $P=P^\tp$ and $\lambda \ge 0$ such that \eqref{eq:explmi} holds, as required.
\else
Omitted. A similar result is proven in \cite{seiler}.
\fi
\end{proof}
With the advent of fast interior-point methods to solve LMIs, the feasibility of the LMI \eqref{eq:explmi} can be quickly ascertained for any fixed $\rho^2$. Since the size of the LMI is on the order of the size of the system $G$ and the IQC $\Pi$, most practical linear systems lead to relatively small LMIs.

Finding the best upper bound amounts to minimizing~$\rho^2$ subject to~\eqref{eq:explmi} being feasible. This type of problem occurs frequently in robust control and is known as a \emph{generalized eigenvalue optimization problem} (GEVP)~\cite{boyd_linear_1997}. The GEVP is not an LMI because~\eqref{eq:explmi} is not jointly linear in $\rho^2$ and $P$. One simple approach to solving the GEVP is to perform a bisection search on $\rho^2$, but there are more sophisticated methods available; see for example~\cite{boyd_elghaoui}.

\begin{rem}
These results may also be carried through in continuous time. In that case, an equation analogous to \eqref{eq:expfdi} must be satisfied for $G(s-\lambda)$ for all $\omega \in [0,\infty)$, and can be verified by finding $P=P^\tp$, $\lambda\ge 0$ such that
\begin{equation*}
  \begin{bmatrix}  
  \hat A^\tp P+ P \hat A-2\lambda P & P\hat B\\ \hat B^\tp P & 0
  \end{bmatrix} + 
  \lambda
  \begin{bmatrix}
  \hat C^\tp \\ \hat D^\tp 
  \end{bmatrix}
  M
  \begin{bmatrix}
  \hat C & \hat D
  \end{bmatrix}
   \prec 0
\end{equation*}
\end{rem}

\section{Examples}\label{sec:examples}
In this section, we show some classes of nonlinearities that can be described by $\rho$-IQCs and therefore used in Theorem~\ref{thm:main} to prove robust exponential stability of an interconnected system. In the case where $\rho=1$, these $\rho$-IQCs reduce to standard IQCs~\cite{megretski_system_1997}. This class of IQCs will be constructed for SISO systems, but they may be adapted for square MIMO systems where the nonlinearity is of the form $\diag(\{\Delta\})$ for a scalar $\Delta$, with little modification.

\subsection{Pointwise IQCs}
A nonlinearity $\Delta$ satisfies a pointwise IQC with a factorization $(\Psi, M)$ if $z_k^\tp Mz_k \geq 0$ for each $k$. In other words, the IQC holds pointwise in time. In this case, $\Delta$ also satisfies the associated $\rho$-IQC for all $\rho \leq 1$. 
Examples of pointwise IQCs include the $\gamma$ \emph{norm-bounded IQC}
\begin{equation*}
  \Pi = \begin{bmatrix}
  \gamma^2&0\\0&-1
  \end{bmatrix}\:,
\end{equation*}
and the \emph{$[\alpha,\beta]$ sector bounded IQC}, given by
\begin{equation*}
  \Pi = \begin{bmatrix}
  -2\alpha\beta&\alpha+\beta\\\alpha+\beta&-2
  \end{bmatrix}\:.
\end{equation*}
Note that the norm-bounded IQC is a special case of the sector IQC with the sector $[-\gamma,\gamma]$. These IQCs hold even if $\Delta$ is time-varying.

\subsection{Zames-Falb IQCs}

A nonlinearity $\Delta$ is \emph{slope-restricted on $[\alpha,\beta]$} where $0 \le \alpha \le \beta$ if the following relation holds for all $x$, $y$.
\begin{equation*}
\bl( \Delta(x)-\Delta(y)-\alpha(x-y) \br)^\tp
\bl( \Delta(x)-\Delta(y)-\beta(x-y) \br) \leq 0\:.
\end{equation*}
This relation states that the chord joining input-output pairs of $\Delta$ has a slope that is bounded between $\alpha$ and $\beta$.
This class of functions satisfies the so-called Zames-Falb family of IQCs \cite{heath_zames-falb_2005,zames_stability_1968}. We give the definition below.

\begin{prop}\label{prop:zf}
A nonlinearity $\Delta$ that is static and slope-restricted on $[\alpha,\beta]$ satisfies the Zames-Falb IQC
\begin{equation}\label{eq:zf}
  \Pi = \addtolength{\arraycolsep}{0mm}
  \bmat{-\alpha\beta(2\!-\!H \!-\! H^*) & \alpha(1\!-\!H) \!+\! \beta(1\!-\!H^*) \\
  \alpha(1\!-\!H^*) \!+\! \beta(1\!-\!H) &
  -(2\!-\!H\!-\!H^*)}
\end{equation}
where $H(z)$ is any proper transfer function with impulse response $h \defeq (h_0,h_1,\dots)$ that satisfies $||h||_1 \le 1$ and $h_k\geq 0$ for all $k$.
\end{prop}
\begin{proof}
See for example~\cite{heath_zames-falb_2005}.
\end{proof}

\begin{rem}\label{remb}
The Zames-Falb IQC~\eqref{eq:zf} admits the factorization
\[
\Psi = \bmat{ \beta(1 - H) & -(1-H) \\ -\alpha & 1 }
\quad\text{and}\quad
M = \bmat{0 & 1 \\ 1 & 0 }
\]
\end{rem}
In general, for a given fixed $\rho$, only a subset of the Zames-Falb IQCs will be $\rho$-IQCs. We now give a characterization of this subset.

\begin{thm}[Zames-Falb $\rho$-IQC]\label{thm:rho_zf}
	Suppose $\Delta$ is static and slope-restricted on $[\alpha,\beta]$. Then $\Delta\in\IQC(\Pi(z),\rho)$ where $\Pi$ is the Zames-Falb IQC~\eqref{eq:zf} and $H$ satisfies the additional constraint
	\begin{equation}\label{eq:rhozf_constraint}
	\sum_{k=0}^\infty \rho^{-2k} h_k \le 1
	\end{equation}
\end{thm}
\begin{proof}
\if\MODE1
The proof involves rewriting the IQC as a discrete-time sum which can be split into parts that can separately be shown to be nonnegative. See the Appendix for the full proof.
\else
Omitted. See~\cite{boczar_arxiv} for a detailed proof.
\fi
\end{proof}

\subsection{Multiple IQCs}
Much like how multiple IQCs can give more precise $L_2$ gain bounds, multiple $\rho$-IQCs can give more precise convergence rates. We present numerical examples with both pointwise and dynamic $\rho$-IQCs.
Consider a stable discrete-time LTI system $G(z)$ in feedback with the sigmoidal nonlinearity $\Delta(x) = b\arctan(x)$. This interconnection is shown in Fig.~\ref{fig:lure_ex}.
\begin{figure}[th]
  \centering
  \includegraphics{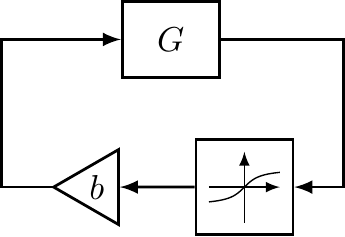}
  \caption{LTI system $G$ in feedback with the static nonlinearity \mbox{$\Delta(x)=b \arctan(x)$}.}\label{fig:lure_ex}
\end{figure}

Since this nonlinearity is static, in the $[0,b]$ sector, and $[0,b]$ slope-restricted, it satisfies the following $\rho$-IQCs
\if\bugfix0
\begin{align}
\Pi_\textup{n}(z) &\defeq \bmat{b^2&0\\0&-1}  && \text{(norm-bounded)} \label{aa}\\
\Pi_0(z) &\defeq \bmat{0&b\\b&-2} && \text{(sector bounded)} \label{bb}\\
\Pi_k(z) &\defeq
\bmat{0&\star \\ b(1\!-\!\rho^{2k} z^{-k}) & -2} && \text{(off-by-$k$ Z-F)} \label{cc}
\end{align}
where the $\star$ makes the matrix Hermitian and we can choose any $k=1,2,\dots$. 
\else
	\if\MODE1
	\begin{align}
	\Pi_\textup{n}(z) &\defeq \bmat{b^2&0\\0&-1}  && \text{(norm-bounded)} \label{aa}\\
	\Pi_0(z) &\defeq \bmat{0&b\\b&-2} && \text{(sector bounded)} \label{bb}\\
	\Pi_k(z) &\defeq
	\bmat{0&b(1 - \rho^{2k} \bar z^{-k}) \\ b(1 - \rho^{2k} z^{-k}) & -2 +\rho^{2k}(z^{-k}+\bar z^{-k})} && \text{(off-by-$k$ Zames-Falb)} \label{cc}
 	\end{align}
	\else
	\begin{align}
	\Pi_\textup{n}(z) &\defeq \bmat{b^2&0\\0&-1}  &&\hspace{-4cm}\text{(norm-bounded)} \label{aa}\\
	\Pi_0(z) &\defeq \bmat{0&b\\b&-2} && \hspace{-4cm}\text{(sector bounded)} \label{bb}\\
	\Pi_k(z) &\defeq
	\bmat{0 & b(1-\rho^{2k} \bar z^{-k}) \\
	b(1-\rho^{2k} z^{-k}) & -2 +\rho^{2k}(z^{-k}+\bar z^{-k})} \notag\\
	& && \hspace{-4cm}\text{(off-by-$k$ Zames-Falb)} \label{cc}
	\end{align}
	\fi
	where we can choose any $k=1,2,\dots$. 
\fi

\paragraph{A tight bound.} For our first example, we analyzed the following LTI system\footnote{This example was inspired by the continuous time example given in \cite{scherer}, which showed that adding more IQCs yields better $L_2$ gain bounds.}
\begin{equation}\label{G1}
G_1(z) \defeq -\frac{(z+1)(10z+9)}{(2z-1)(5z-1)(10z-1)}
\end{equation}
We solved the feasibility LMI~\eqref{eq:explmi} using MATLAB together with the CVX package~\cite{cvx2,cvx} to find the fastest guaranteed rate of convergence. We searched over positive linear combinations of subsets of the IQCs~\eqref{aa}--\eqref{cc}. Fig.~\ref{fig:rates} shows the rate bounds achieved as a function of which IQCs were used. For the particular choice $b=1$, Fig.~\ref{fig:states} shows sample state trajectories.

The true exponential rate can be found by linearizing the system about its equilibrium point. Namely, $\Delta(x) \approx bx$. Formally, this is an application of Lyapunov's indirect method~\cite[Thm.~4.13]{khalil}. The result is that the decay rate should correspond to the maximal pole magnitude of the closed-loop map $G(z)/(1-bG(z))$. We display the true exponential rate as the dashed black curve in Fig.~\ref{fig:rates} and Fig.~\ref{fig:states}.

For this example, the $\rho$-IQC approach yields a tight upper bound to the true exponential rate when we use a combination of the sector and off-by-1 IQCs.
\begin{figure}[th]
  \centering
  \includegraphics{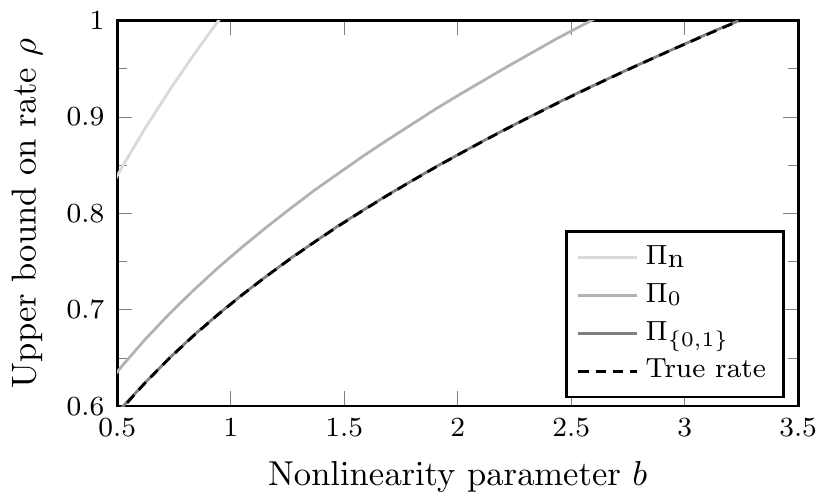}
  \caption{Upper bounds on the exponential convergence rate $\rho$ for the system $G_1(z)$ given in~\eqref{G1} in feedback as in Fig.~\ref{fig:lure_ex}. A tight bound is achieved using two $\rho$-IQCs.\label{fig:rates}}
\end{figure}

\begin{figure}[th]
  \centering
  \includegraphics{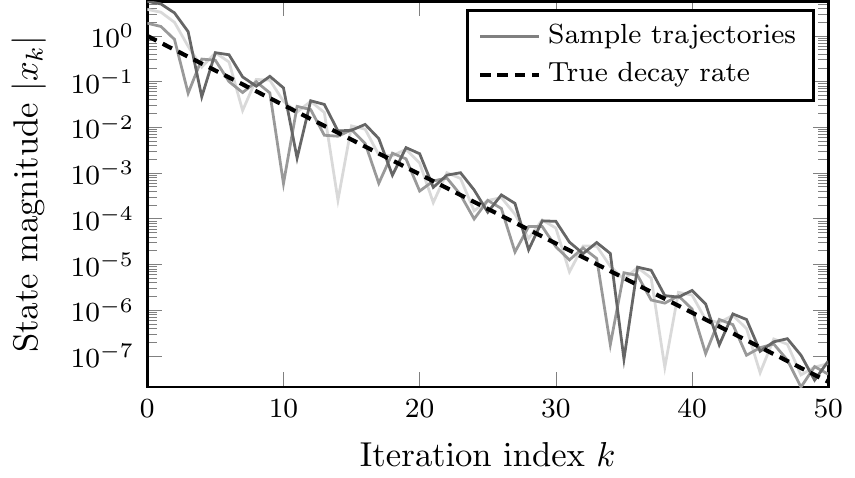}
  \caption{State decay over time of the system $G_1(z)$ in feedback as in Fig.~\ref{fig:lure_ex} with $b=1$ for various initial conditions $x_0\in[-15,15]$. The dashed black line is $\rho^k$, where $\rho=.7058$ is the true rate at $b=1$ in Fig.~\ref{fig:rates}.}
  \label{fig:states}
\end{figure}

\paragraph{A loose bound.} The $\rho$-IQC approach does not always achieve tight bounds as in the previous example. Consider the same problem as before but this time using
\begin{equation}\label{G2}
G_2(z) \defeq \frac{2z-1}{10(2z^2-z+1)}
\end{equation}
The rate bounds for various $\rho$-IQCs are shown in Fig.~\ref{fig:rates_ex2}. This time, we again observe that using more IQCs achieves better rate bounds, but the bound is not tight even after using six IQCs.

\begin{figure}[thpb]
	\centering
	\includegraphics{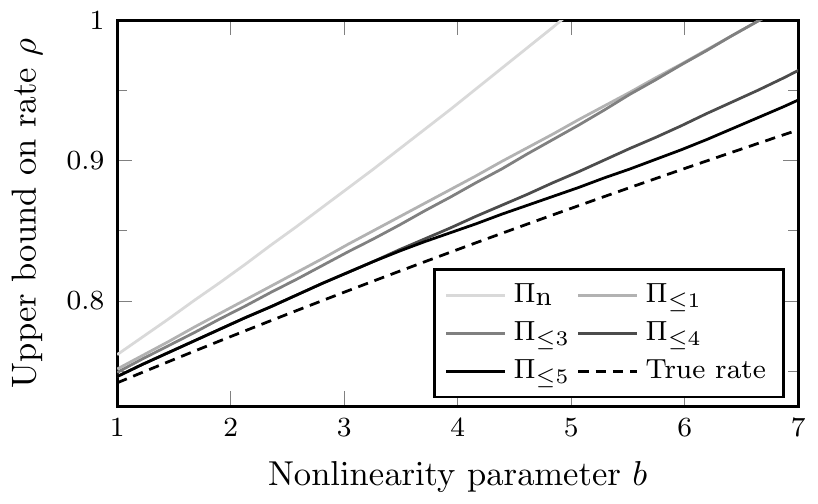}
	\caption{Upper bounds on the exponential convergence rate $\rho$ for the system $G_2(z)$ given in~\eqref{G2} in feedback as in Fig.~\ref{fig:lure_ex}. As we include more $\rho$-IQCs, we can certify tighter bounds.}
	\label{fig:rates_ex2}
\end{figure}

\section{Conclusion}\label{sec:conclusion}

We presented a modification of IQC theory that allows the certification of exponential rates. Although we only gave the $\rho$-IQC specialization for pointwise and Zames-Falb IQCs, the concept can in principle be extended to other IQCs such as, for example, uncertain time delays or slowly varying systems~\cite{megretski_system_1997}. As the dictionary of $\rho$-IQCs is further populated, the applicability of the technique outlined herein would be correspondingly expanded.

\section{Acknowledgments}\label{sec:acknowledgments}

The authors would like to thank Andrew Packard and \mbox{Murat Arcak} for very helpful discussions. L.~Lessard is supported by AFOSR award FA9550-12-1-0339.


\if\MODE1
\appendix

\section{Appendix}

\subsection{Proof of Proposition~\ref{prop:exp}}

Suppose the interconnection of Fig.~\ref{fig:lure2} is stable. Then there exists some $K>0$ such that for any choice of the signals $e$ and $f$ and for all $T$,
\begin{equation}\label{hh}
\sum_{k=0}^T  \bl(\|w_k\|^2 + \|v_k\|^2\br) \leq  K \sum_{k=0}^T  \bl(\|e_k\|^2 + \|f_k\|^2\br)
\end{equation}
The proof will follow by carefully choosing $e$ and $f$ to transform Fig.~\ref{fig:lure2} into Fig.~\ref{fig:lure1}.
To this end, note that $(A,B)$ is controllable by assumption. So there exists a finite sequence of inputs $u_0,\dots,u_{n-1}$ and corresponding outputs $y_0,\dots,y_{n-1}$ that drives the state of $G$ from $\xi_0=0$ to $\xi_n=x_0$. Therefore, if we set
\[
e_k = \begin{cases}\rho^{-k}u_k & 0\le k < n\\ 0 & k \ge n\end{cases}
,\;\;\;
f_k = \begin{cases}-\rho^{-k}y_k & 0\le k < n\\ 0 & k \ge n\end{cases}
\]
then we obtain $\xi_n = x_0$ in the interconnection of Fig.~\ref{fig:lure2}. Moreover,  $f_k=e_k=0$ for $k\ge n$, so we may cancel the $\rho$ blocks using the fact that $\rho_-\circ\rho_+$ is the identity operator. It follows that for  $k\ge n$, the two interconnections become identical and therefore $\xi_k = x_{k-n}$.

Substituting into~\eqref{hh}, we conclude that
\begin{equation}\label{hhh}
\sum_{k=0}^T  \bl(\|w_k\|^2 + \|v_k\|^2\br) \leq  K \sum_{k=0}^{n-1} \bl(\|e_k\|^2 + \|f_k\|^2\br)
\end{equation}
The right-hand side of~\eqref{hhh} is independent of $T$, but~\eqref{hhh} holds for all $T$ so we must have
\begin{equation*}
\lim_{k \to\infty} \|w_k\| = 0
\qquad\text{and}\qquad
\lim_{k \to\infty} \|v_k\| = 0
\end{equation*}
For $k\ge n$, we have $w_k = \rho^{-k}u_k$ and $v_k = \rho^{-k} y_k$. Therefore there exists some constant $c>0$ such that
\[
\|u_k\| \le c \rho^k
\qquad\text{and}\qquad
\|y_k\| \le c \rho^k
\]
Now $(A,C)$ is observable by assumption, so let $L$ be such that the eigenvalues of $A+LC$ are all zero. Rewrite the dynamics of $G$ as
\begin{equation}\label{gg}
x_{k+1} = \bar A x_k +\bar B h_k
\end{equation}
where $\bar A \defeq A+LC$, $\bar B \defeq  \bmat{LD+B & -L}$, and $h_k \defeq \bmat{u_k^\tp  & y_k^\tp }^\tp $.
Iterating~\eqref{gg}, we obtain
\begin{equation}\label{ggg}
x_k = \bar A^k x_0 + \sum_{i=0}^{k-1} \bar A^{k-1-i} \bar B h_i
\end{equation}
Since all eigenvalues of $\bar A$ are zero, $\bar A$ is nilpotent and so $\bar A^n = 0$. For $k\ge n$, \eqref{ggg} therefore becomes
\begin{equation*}\label{qq}
x_k = \sum_{i=0}^{n-1} \bar A^{n-1-i} \bar B h_{k-n+i}
\end{equation*}
We can now bound the state using the triangle inequality.
\begin{align*}
\|x_k\|
&\le \underbrace{\bl\| \bmat{ \bar A^{n-1}\bar B & \dots & \bar A \bar B & \bar B} \br\|}_\gamma \sum_{i=k-n}^{k-1} \|h_i\| \\
&\le \gamma\, c \left( \frac{ \rho^{-n}-1}{1-\rho} \right) \rho^k
\end{align*}
and this completes the proof.\qedhere

\subsection{Proof of Theorem~\ref{thm:rho_zf}}

We will prove this general result by first considering the simpler case where the slope restriction is on $[\alpha,\beta] = [0,\infty]$ and $H(z) = \rho^{2j} z^{-j}$. Note that this choice trivially satisfies~\eqref{eq:rhozf_constraint}. In this case, the~$\Pi$ from~\eqref{eq:zf} becomes
\begin{equation}\label{obj}
\Pi = \bmat{ 0 & 1-\rho^{2j}\bar z^{-j} \\ 1-\rho^{2j}z^{-j} & 0 }
\end{equation}
where $\bar z$ denotes the complex conjugate of $z$. We call~\eqref{obj} the ``off-by-$j$'' Zames-Falb IQC. We would like to show that $\Delta \in \IQC(\Pi(z),\rho)$. Appealing to Definition~\ref{def:piqc} and Remarks~\ref{rem:rhoIQC_time_domain} and~\ref{remb}, this amounts to proving that
\begin{equation}\label{kk2}
\sum_{k=0}^\infty \rho^{-2k} u_k^\tp ( y_k - \rho^{2j} y_{k-j} ) \ge 0
\end{equation}
We will prove~\eqref{kk2} by borrowing the approach from~\cite{lessard_analysis_2014}. If $\Delta$ is multidimensional, we require that $\Delta$ be the gradient of a potential function~\cite{heath_zames-falb_2005}. By the assumption that $\Delta$ is slope-restricted on $[0,\infty]$, we have
\[
(\Delta(x)-\Delta(y))^\tp(x-y) \ge 0 \quad\text{holds for all }x,y
\]
In other words, $\Delta$ is monotone. Now define the scalar function $g$ such that $\grad g = \Delta$. By Kachurovskii's theorem, $g$ is convex and we have
\[
g(y) \ge g(x) + \Delta(x)^\tp(y-x)
\quad\text{for all }x,y
\]
Moreover, setting $(x,y)\mapsto(y_k,0)$ or $(x,y)\mapsto(y_k,y_{k-j})$ leads to the two inequalities:
\begin{align}
u_k^\tp y_k &\ge g(y_k) \label{f1} \\
u_k^\tp (y_k - y_{k-j}) &\ge g(y_k) - g(y_{k-j}) \label{f2}
\end{align}
We will assume for simplicity that $g(x) \ge 0$ for all $x$.
Substituting~\eqref{f1} and \eqref{f2} into the left-hand side of~\eqref{kk2}, the partial sum from $0$ to $T$ is:
\begin{align*}
& \sum_{k=0}^T \rho^{-2k} u_k^\tp ( y_k - \rho^{2j} y_{k-j} ) \\
&=\sum_{k=0}^T \rho^{-2k}\bl( (1-\rho^{2j})u_k^\tp y_k + \rho^{2j}u_k^\tp(y_k-y_{k-j})\br) \\
&\ge \sum_{k=0}^T \rho^{-2k}\bl( (1-\rho^{2j})g(y_k) + \rho^{2j}(g(y_k)-g(y_{k-j})) \br) \\
&= \sum_{k=0}^T \rho^{-2k}\bl( g(y_k) - \rho^{2j}g(y_{k-j}) \br) \\
&= \sum_{k=T-j+1}^{T} \rho^{-2k} g(y_k) \ge 0
\end{align*}
Since each partial sum is nonnegative, the infinite sum (which must converge) is also nonnegative, and therefore we have proven~\eqref{kk2}.
For the case of general $[\alpha,\beta]$, observe that $\Pi$ from~\eqref{eq:zf} may be factored as
\[
\Pi = \bmat{ \beta & -1 \\ -\alpha & 1 }^\tp
\bmat{ 0 & 1-H(z)^* \\ 1-H(z) & 0 }  
\bmat{ \beta & -1 \\ -\alpha & 1 }
\]
So the above proof holds verbatim if we make the substitution $(y,u) \mapsto (\beta y - u,u-\alpha y)$.

Now we consider the case of a more general $H(z)$. Suppose $H(z) = \sum_{k=0}^\infty h_k z^{-k}$ where $h_k$ satisfies~\eqref{eq:rhozf_constraint}. Then,
\begin{align*}
1- H(z) &= 1-\sum_{k=0}^\infty h_k z^{-k}\\
&= \underbrace{1-\sum_{k=0}^\infty \rho^{-2k}h_k}_c + \sum_{k=0}^\infty\rho^{-2k}h_k\left(1-\rho^{2k}z^{-k}\right) \\
&=  c\,(1-H_s) + \sum_{k=0}^\infty\rho^{-2k}h_k\left(1-H_k(z)\right)\:,
\end{align*}
where $H_k(z)=\rho^{2k}z^{-k}$ and $H_s=0$. Note that $H_k(z)$ corresponds the off-by-$k$ Zames-Falb IQC, which we proved above is a $\rho$-IQC. Also, $H_s$ corresponds to the sector IQC, which is also a $\rho$-IQC. Now note that the general Zames-Falb IQC~\eqref{eq:zf} is linear in $I-H$ and $I-H^*$. Therefore, $\Pi(z)$ is a positive linear combination of $\rho$-IQCs and must therefore be a $\rho$-IQC itself. \qedhere
\else\fi
{
\if\MODE1\else\small\fi
\bibliographystyle{abbrv}
\bibliography{freq}
}

\end{document}